\documentclass[10pt,conference]{IEEEtran} 
\IEEEoverridecommandlockouts
\usepackage[linesnumbered,vlined,ruled]{algorithm2e}
\usepackage{algorithmic}
\usepackage{etoolbox}
\usepackage{cite}
\usepackage{graphicx}
\graphicspath{ {./images/} }
\usepackage{amsmath,amsthm,amssymb,amsfonts}
\usepackage{mathtools}
\usepackage[dvipsnames]{xcolor}
\usepackage{dcolumn}
\usepackage[utf8]{inputenc}
\usepackage{soul}
\usepackage{array}
\usepackage{tabulary}
\usepackage{enumitem}

\newtheorem{definition}{Definition}   %
\theoremstyle{definition}             
\newtheorem{policy}{Policy}           
\theoremstyle{remark}                 
\theoremstyle{plain}                  
\newtheorem{lemma}{Lemma}             

\newcommand{\comments}[1]{{\leavevmode\color{black}#1}}
\newcommand{\tann}[1]{\textcolor{black}{#1}}

\newcommand{\define}{\triangleq}

\newcommand{\vecG}{\boldsymbol}
\renewcommand{\vec}{\mathbf}
\DeclarePairedDelimiter{\set}{\{}{\}}

\DeclarePairedDelimiter{\Inorm}{\|}{\|_1}
\DeclarePairedDelimiter{\Paren}{\bigg(}{\bigg)}
\DeclarePairedDelimiter{\Bracket}{\bigg[}{\bigg]}
\DeclarePairedDelimiter{\Brace}{\bigg\{}{\bigg\}}

\newcommand{\deny}[1]{}
\newcommand{\delete}[2]{}

\newcommand{\AP}{\dagger}
\newcommand{\ES}{\ddagger}
\newcommand{\apSet}{\mathcal{K}}
\newcommand{\esSet}{\mathcal{M}}
\newcommand{\ccSet}{\mathcal{X}}
\newcommand{\jSpace}{\mathcal{J}}
\newcommand{\Stat}{\mathbf{S}}

\newcommand{\Policy}{\vecG{\Omega}}
\newcommand{\Delay}{\vecG{\mathcal{D}}}
\newcommand{\Baseline}{\vecG{\Pi}}

\newcommand{\brlatency}{signaling latency}

\begin{document}
    \title{
        Online Distributed Job Dispatching with Outdated and Partially-Observable Information
    }

    \author{
        \IEEEauthorblockN{
            Yuncong Hong\IEEEauthorrefmark{2}\IEEEauthorrefmark{3}\IEEEauthorrefmark{4},
            Bojie Lv\IEEEauthorrefmark{2}\IEEEauthorrefmark{4},
            Rui Wang\IEEEauthorrefmark{2}\IEEEauthorrefmark{4},
            Haisheng Tan\IEEEauthorrefmark{1}\IEEEauthorrefmark{4},
            Zhenhua Han\IEEEauthorrefmark{3},
            Hao Zhou\IEEEauthorrefmark{1},
            Francis C.M. Lau\IEEEauthorrefmark{3}
        }
        \IEEEauthorblockA{\IEEEauthorrefmark{1}
            LINKE Lab, University of Science and Technology of China, Hefei, China \\
        }
        \IEEEauthorblockA{\IEEEauthorrefmark{2}
            Department of Electrical and Electronic Engineering, Southern University of Science and Technology, Shenzhen, China \\
        }
        \IEEEauthorblockA{\IEEEauthorrefmark{3}
            Department of Computer Science, The University of Hong Kong, Hong Kong, China \\
        }
        \IEEEauthorblockA{\IEEEauthorrefmark{4}
            Research Center of Networks and Communications, Peng Cheng Laboratory, Shenzhen, China \\
        }
    }%


    \maketitle
    
    \begin{abstract}
    In this paper, we investigate  online distributed job dispatching  in an edge computing system residing in a Metropolitan Area Network (MAN).
    Specifically, job dispatchers are implemented on access points (APs) which collect jobs from mobile users and distribute each job to a server at the edge or the cloud. A signaling mechanism with periodic broadcast is introduced to facilitate cooperation among APs.
    The transmission latency is non-negligible in MAN, which leads to outdated information sharing among APs. 
    Moreover, the fully-observed system state is discouraged as reception of all broadcast is time consuming. Therefore, we formulate the distributed optimization of job dispatching strategies among the APs as a Markov decision process with partial and outdated system state, \emph{i.e.}, partially observable Markov Decision Process (POMDP). The conventional solution for POMDP is impractical due to huge time complexity.  We propose a novel low-complexity solution framework for distributed job dispatching, based on which the optimization of job dispatching policy can be decoupled via an alternative policy iteration algorithm, so that the distributed policy iteration of each AP can be made according to partial and outdated observation. A theoretical performance lower bound is proved for our approximate MDP solution. Furthermore, we conduct extensive simulations based on the Google Cluster {trace}. {The evaluation results show that our policy can achieve {as high as} $20.67\%$ reduction in average job response time compared with heuristic baselines, and our algorithm consistently performs well under various parameter settings.}
\end{abstract}


    \section{Introduction}
\label{sec:introduction}
Edge computing is a promising solution for increasing computation-intensive and energy-hungry applications on mobile devices.
Large amount of mobile devices can connect to the access points (APs) which function as gateways to aggregate and dispatch jobs to the edge servers \cite{MEC-SURVEY}.
The edge servers are deployed in closer proximity to APs than cloud infrastructures, which alleviate the communication overhead and enable the computation of time-sensitive jobs.
However, the edge servers are usually deployed with limited computation resources.
The establishment of efficient cooperation among edge servers is one of the major design challenges, given the signaling overhead and latency of distributed information exchange and decision making.

\begin{figure}[htp!]
    \centering
    \includegraphics[width=0.42\textwidth]{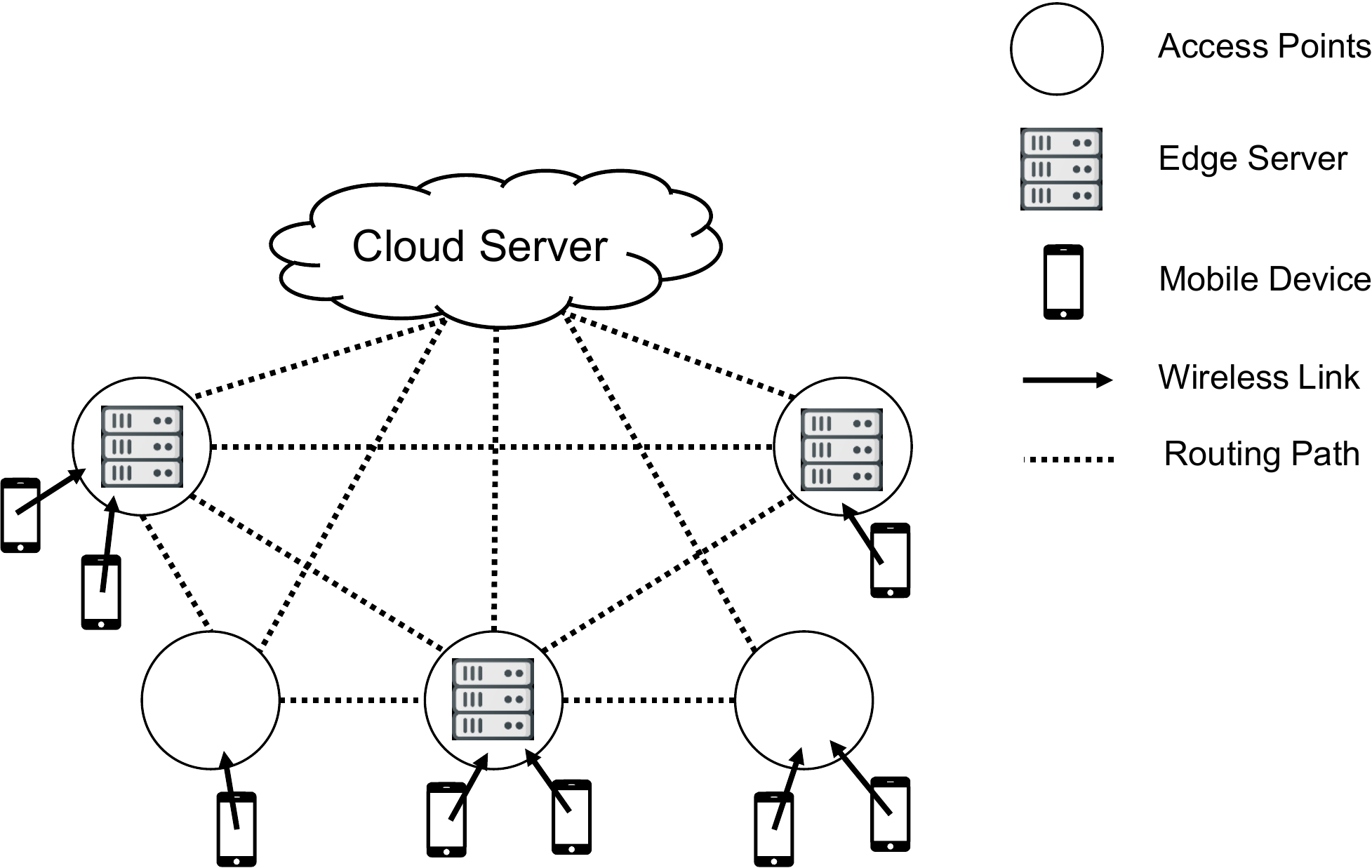}
    \caption{The Illustration of System Model}
    \label{fig:system}
\end{figure}

We consider an edge computing system with multiple APs and edge servers residing in the Metropolitan Area Network (MAN) as illustrated in Fig.\ref{fig:system}.
The APs collect jobs offloaded from the mobile users in its service area and make dispatching decision for each job. According to the MAN performance analysis in \cite{MAN-LATENCY}, the data transmission latency {among APs and edge servers} varies a lot with respect to different hours of day and devices' locations in a MAN. {In addition, each AP also suffers from signaling latency, which is the time consumed for each AP to collect system state information under some signaling mechanism.}
There has been a number of existing literature considering random transmission latency of job delivery in edge computing networks (e.g., \cite{latency-EDGE19,MOBIHOC19-ZhouZ,IOTJ18-FanQ,TOC19-LiuC,JSAC19-AlameddineHA}). 
\tann{However, there are only a few works considering signaling latency of cooperation among distributed job dispatchers \cite{JSAC17-LyuX,TWC18-LyuX}.}
In fact, it is full of challenge to consider latency in {both job delivery and} signaling. The cooperation of distributed dispatchers suffers from significant unpredictable signaling and transmission overhead. Therefore, aggregating the global system state information at each AP is not practical, and a centralized dispatcher design is discouraged. Moreover, the latency will also lead to outdated information at each dispatcher and information inconsistency among different dispatchers\comments{, which may introduce ineliminable estimation error on the number of jobs in the system and thus discourages the cooperative dispatcher design}.

In this paper, we would like to shed some lights on the above challenging distributed dispatcher design via proposing the POMDP (partially observable Markov decision process) problem formulation and a novel low-complexity approximate MDP solution framework.
Specifically, we consider a practical scenario where the signaling latency among APs and edge servers as well as job uploading latency from APs to edge servers are assumed to be random, and each AP can only receive the broadcast system state information from part of the APs with random latency (i.e., the global system state is not available at the APs).
Our main contributions in this new optimization scenario are summarized as follows.
\begin{itemize}
    \item 
    We propose a novel low-complexity distributed solution framework for the dispatcher design at APs.
    By leveraging partial observation at each dispatcher, we derive the expression of approximate value function under MDP framework and decouple the optimization problem onto each AP.
    Thus, the complicated POMDP solution is avoided.
    To our best knowledge, this is the first work to address the cooperative distributed multi-agent optimization problem with outdated and partial information under MDP framework.
    \item 
    We derive an analytical cost lower bound for the proposed distributed dispatching policy in the above low-complexity solution framework.
    In the conventional approximate MDP methods, the solution is usually evaluated via numerical method which are hard to obtain analytical performance bound.
    \item We conduct extensive simulations based on the Google Cluster trace, compared with three heuristic benchmarks. The evaluation results show that our proposed job dispatching policy can achieve $20.67\%$ reduction in average job response time, and our algorithm consistently performs well under various parameter settings of signaling latency, job arrival intensity and job processing time.
\end{itemize}

\noindent \textbf{Paper Organization: }The remainder of this paper is organized as follows.
In Section \ref{sec:review}, the related works are elaborated.
In Section \ref{sec:model}, we illustrate the system model and the signaling model with random latency.
In Section \ref{sec:formulation}, we formulate the global optimization of dispatching decisions at all APs as an POMDP.
In Section \ref{sec:algorithm}, we introduce the novel low-complexity distributed solution framework for the above POMDP.
The numerical analysis of the proposed solution is provided in Section \ref{sec:evaluation}, and the conclusion is drawn in Section \ref{sec:conclusion}.

\section{Related Work}
\label{sec:review}

There have been a number of works considering the centralized job dispatching with updated and complete knowledge on the system states of edge computing systems.
For example, in order to minimize the average job response time in the worst case, the authors in \cite{tan-online} designed an online algorithm for job dispatching in edge computing systems with fixed uploading latency.
In the scenario that APs and edge servers are connected via software defined network (SDN), the authors in \cite{IOTJ18-FanQ} proposed a heuristic algorithm to dispatch the jobs to the closest edge servers according to their locations.
Considering random jobs arrival and job offloading to a single edge server, the authors in \cite{mdp-globecom,mdp-tvt} formulate the offloading problem as an infinite-horizon Markov decision process (MDP).
In the above works, a centralized dispatcher with complete and updated knowledge of the system states was assumed in the edge computing systems, which might be impractical.

Hence, there are also some works considering the distributed job dispatching in edge computing systems.
For example, in order to minimize a weighted sum of total energy consumption and uploading latency, the authors in \cite{ToN-Xuchen2016} proposed a distributed job dispatching algorithm based on game theory to achieve the Nash equilibrium. 
Considering job migration at edge servers, the authors in \cite{ToN-xujie2018} optimized the edge computing performance in a distributed manner with limited energy resources via a congestion game framework.
In the scenario that APs cooperatively dispatch jobs with multiple edge servers, the authors in \cite{mdp-jcin} proposed a novel approximate MDP solution framework to alleviate the algorithm complexity and minimize the average job response time.
However, in the above works, the latency of information exchange among APs and edge servers is ignored.
In fact, due to the complicated network traffic, this latency might be significant, and the staleness of system state information at the dispatcher of a edge computing systems should be considered.

The staleness of information sharing among APs and edge servers may degrade the performance of the job dispatching algorithm in edge computing systems.
To the best of our knowledge, there are very limited works investigating this issue.
For example, the authors in \cite{JSAC17-LyuX} proposed a randomized policy via Lyapunov optimization approach to stabilize the queues in a MEC system with multiple IoT devices offloading jobs to one edge server, where \brlatency~is considered. 
In \cite{TWC18-LyuX}, the above approach was applied to the scenario that mobile devices offload jobs to each other via D2D link.
In the above two works, there is one centralized dispatcher in the system and the objective is to stabilize the transmission queues.
Hence, the existence of \brlatency~may not raise significant challenge to the algorithm designs.

However, the design of distributed dispatchers with \brlatency~could be more challenging.
For example, the signaling latency at distributed dispatchers could be different, and the synchronization of their dispatching decisions become infeasible.
Furthermore, taking the signaling overhead in consideration, it is of more practical significance favor for the distributed dispatchers to make decisions based on locally observed system state information, instead of global system state information.
To our best knowledge, there is no appropriate optimization framework for the distributed dispatcher design with both \brlatency~and partially observable system state information to date.



    \section{System Model}
\label{sec:model}

\subsection{Network Model}
We consider an edge computing system with $K$ Access Points (APs), $M$ edge servers, and $1$ cloud server as illustrated in Fig.\ref{fig:system}.
The sets of APs and processing servers are denoted as $\apSet \define \set{1,\dots,K}$ and $\esSet \define \set{0,\dots,M}$, respectively.
Specifically, we have \comments{\emph{processing server} denote both edge servers and the cloud server for job processing} and the $0$-th processing server denote the cloud server.
\comments{One edge server is always collocated with one AP, i.e., the edge server is deployed at the same place with the access point.}
Each AP collects the jobs from the mobile users within its coverage, and makes dispatching action on the processing servers for each job.
Furthermore, it is assumed that one AP could only access to the cloud server, the collocated edge server if existed, and neighbor edge servers \comments{(i.e., the pair of edge servers connected via the routing path as illustrated in Fig.\ref{fig:system})} due to transmission latency limitation.
We refer to $\esSet_{k} \subseteq \esSet$ as the \emph{candidate server set} of the $k$-th AP,  $\apSet_{m} \subseteq \apSet$ as the \emph{potential AP set} of the $m$-th edge server, and $\rho_{k,m}$ as the collocation indicator (i.e., $\rho_{k,m}=1$ if $k$-th AP and $m$-th edge server collocated, otherwise $\rho_{k,m}=0$) ($\forall k\in\apSet, m\in\esSet$).
We consider cloud server as one special processing server with stronger computation capability, but also with larger transmission latency compared with the edge servers.

The dispatchers are implemented on the APs in a distributed manner.
Without loss of generality, it is assumed that there are $J$ types of jobs computed in this system, which are denoted via the set $\jSpace \define \set{1,\dots,J}$.
The time axis is organized by time slots.
The arrivals of the type-$j$ jobs at the $k$-th AP in different time slots are assumed to be independent and identically distributed Bernoulli random variables, and the arrival probability are denoted as $\lambda_{k,j}$.
Each AP immediately dispatches each type of arrived jobs to one processing server.
\comments{It's assumed that the network traffic over the routing path is random, and the job uploading from one AP to one edge server consumes a random number of time slots.}
Let $\mathbb{U}_{k,m,j}(\Xi)$ be the uploading latency distribution of the type-$j$ jobs from the $k$-th AP to the $m$-th server with finite support $\set{1, \dots, \Xi}$, whose expectation is denoted as $u_{k,m,j}$.
Specifically, the uploading latency is fixed as $0$ if $\rho_{k,m}=1$ for job uploading to the collocated edge server ($\forall k\in\apSet,m\in\esSet,j\in\jSpace$).

We adopt the \emph{unrelated machines assumption} as in \cite{tan-online} for job computation process, where the computation time on different processing servers would follow independent distribution.
Specifically, there are $J$ parallel virtual machines (VMs) running on each processing server for the $J$ job types, respectively.
It is assumed that the computation time of different job types on different edge servers follows independent memoryless geometric distribution 
with different expectations as in \cite{TOWC18-HuangKb}.
Let $\mathbb{G}(1/c_{m,j})$ be the distribution of the computation time slots for the type-$j$ jobs on the $m$-th processing server, where $\mathbb{G}$ denotes the geometric distribution, $c_{m,j}$ is the expectation, and {$1/c_{m,j}$ represents the parameter of geometric distribution}.
For each job type, the uploaded jobs are computed in a First-Come-First-Serve (FCFS) manner, and a processing queue with a maximum job number $L_{max}$ is established for each VM.
The arrival jobs will be discarded on processing server when the processing queue is full.
\subsection{{Signaling Mechanism with Periodic Broadcast}}
\label{subsec:broadcast}
In order to facilitate distributed dispatching for the APs, the signaling mechanism with periodic broadcast is introduced.
We refer to every $t_B$ time slots as a broadcast interval.
At the beginning of each broadcast interval, the local state information (LSI) of APs and processing servers are broadcast, and each AP updates its dispatching strategy of job dispatching when observing the broadcast LSIs from some APs and processing servers.
\comments{The contents of the LSI at the APs and processing servers are given in the following definitions.}%

\begin{definition}[LSI of APs]
    Let $R^{(k)}_{m,j}(\xi,t,n) \in \set{0,1}$ be the indicator of the type-$j$ jobs at the $n$-th time slot of the $t$-th interval.
    Its value is $1$ when there is one job being uploaded from the $k$-th AP to the $m$-th server which has been delivered for $\xi$ time slots, and $0$ otherwise ($\forall k\in\apSet, m\in\esSet, j\in\jSpace$).
    Let $\omega_{k,j}(t)$ be the target processing server for the type-$j$ jobs of the $k$-th AP dispatched at the beginning of the $t$-th broadcast interval.
    And the LSI of the $k$-th AP at the beginning of the $t$-th broadcast interval is defined as
    {\small
    \begin{align}
        \mathcal{R}_{k}(t) \define
        \Paren{
            \Brace{\vec{R}^{(k)}_{m,j}(t,0) \Big| \forall m\in\esSet,j\in\jSpace},
            \mathcal{A}_{k}(t)
        },
    \end{align}
    }%
    where
    {\small
    \begin{align}
        \vec{R}^{(k)}_{m,j}(t,0) &\define \Paren{
            R^{(k)}_{m,j}(0,t,0), \dots, R^{(k)}_{m,j}(\Xi,t,0)
        },
        \\
        \mathcal{A}_{k}(t) &\define \Brace{\omega_{k,j}(t) \Big| \forall j\in\jSpace}
    \end{align}
    }%
    are referred as status of the type-$j$ job from the $k$-th AP to the $m$-th processing server, and dispatching actions of the $k$-th AP at the beginning of the $t$-th broadcast interval, respectively. 
\end{definition}

\begin{definition}[LSI of Processing Servers]
    Let $Q_{m,j}({t,n})$ be the number of type-$j$ jobs on the $m$-th processing server at the $n$-th time slot of the $t$-th interval ($\forall m\in\esSet, j\in\jSpace$).
    The LSI of the $m$-th processing server at beginning of the $t$-th broadcast interval is defined as
    {\small
    \begin{align}
        \mathcal{Q}_{m}(t) \define \Brace{
            Q_{m,j}(t, 0) \Big| \forall j\in\jSpace
        }.
    \end{align}
    }%
\end{definition}

\comments{Moreover, we refer to \emph{global state information} (GSI) as the aggregation of LSIs of all the APs and processing servers in one broadcast interval.}%
\begin{definition}[Global State Information]
    At the $t$-th broadcast interval, global state information (GSI) is defined as follows.
    {\small
    \begin{align}
        \Stat(t) \define
            \Paren{
                \Brace{\mathcal{R}_{k}(t) \Big| \forall k\in\apSet},
                \Brace{\mathcal{Q}_{m}(t) \Big| \forall m\in\esSet}
            }.
    \end{align}
    }%
\end{definition}

\comments{As a remark notice that the observable LSI may be outdated due to signaling latency among APs and processing servers.}
As the APs and processing servers may reside in different locations of a MAN, the transmission latency of LSI is not negligible. 
It might be inefficient for one AP (say the $k$-th AP) to collect the complete GSI before updating the dispatching policy.
For example, the transmission latency of the LSI from the edge servers out of its \emph{candidate server set} $\esSet_{k}$ may be large, and some broadcast information may be discarded by the routers after a certain number of hops.
\comments{
    Here, we firstly define \emph{conflict AP set} and notice that only the LSIs of a subset of APs are interested for one single dispatcher, and then we define the interested partially-observable information as \emph{observable state information} (OSI) based on the \emph{conflict AP set} and \emph{candidate server set}.
}%
The definitions of \emph{conflict AP set} and OSI are given below, respectively.
\begin{definition}[Conflict AP Set]
    The conflict AP set to the $k$-th AP consists of the neighboring APs who have direct impacts on the queueing time of the jobs on processing servers dispatched from the $k$-th AP, i.e., $ \ccSet_{k} \define \bigcup_{m\in\esSet_{k}} \apSet_{m}$.
\end{definition}

\begin{definition}[Observable State Information]
    The observable state information (OSI) of the $k$-th AP ($\forall k\in\apSet$) at the $t$-th broadcast interval is defined as the aggregation of LSIs of the APs in {conflict AP set} and the edge servers in {candidate server set} of the $k$-th AP, i.e.,
    {\small
    \begin{align}
        \Stat_{k}(t) &\define
        \Paren{
            \Brace{\mathcal{R}_{k'}(t) \Big| \forall k'\in\ccSet_{k}},
            \Brace{\mathcal{Q}_{m}(t) \Big| \forall m\in\esSet_{k}}
        }.
    \end{align}
    }%
    \label{def:OSI}
\end{definition}


The $k$-th AP is able to collect its OSI $\Stat_{k}(t)$ at the $\mathcal{D}_{k}(t)$-th time slots of the $t$-th broadcast interval, where $\mathcal{D}_{k}(t)$ \comments{denotes the \brlatency~of the $k$-th AP at the $t$-th broadcast interval which is a random variable in the unit of timeslot}.
It is assumed that $\mathcal{D}_{k}(t)$ follows identical and independent distribution in different broadcast interval.


    \section{POMDP-based Problem Formulation}
\label{sec:formulation}
In this section, we formulate the optimization of job dispatching at all APs as a Markov decision process (MDP) problem.
Since each AP updates the job dispatching actions according to OSI instead of GSI, the MDP problem is a partially observable MDP (POMDP).
\comments{
    Firstly, we give the definitions of \emph{dispatching policy} and \emph{cost function}, together with the \emph{system state} (i.e., the GSI) defined previously, to complete the MDP problem formulation.
}%

\begin{definition}[Dispatching Policy]
    The individual dispatching policy of the $k$-th AP, denoted as $\Omega_{k}$ ($\forall k \in\apSet$), maps from its OSI $\Stat_{k}$ and its \brlatency~$\mathcal{D}_{k}$ to the dispatching action for each job type, i.e.,
    {\small
    \begin{align}
        \Omega_{k} \Paren{ \Stat_{k}(t), \mathcal{D}_{k}(t) }
        &\define \mathcal{A}_{k}(t+1)
        \label{def:action}
    \end{align}
    }%

    The aggregation of individual policy of all APs is referred to as the system dispatching policy $\Policy$.
    Thus,
    {\small
    \begin{align}
        \Policy\Paren{ \Stat(t), \Delay(t) } \define \Brace{
            \Omega_{1}(\Stat_{1}(t), \mathcal{D}_{1}(t)), \dots, \Omega_{K}(\Stat_{K}(t),\mathcal{D}_{K}(t))
        },
    \end{align}
    }%
    where $\Delay(t) \define \set{ \mathcal{D}_{1}(t), \dots, \mathcal{D}_{K}(t) }$.
\end{definition}

According to the Little's law \cite{Little1961}, the average response time per job, counting the number of broadcast intervals from job arrival to the accomplishment of computation, is proportional to the number of jobs in the system, given the job arrival rates at all the APs.
Hence,
we define the cost function per broadcast interval\comments{, which consists the cumulative cost to be minimized in the MDP problem,} as follows.

\begin{definition}[Cost Function per Broadcast Interval]
    The cost function of the $t$-th broadcast interval with $\Stat(t)$ is defined as
    {\small
    \begin{align}
        g\Paren{\Stat(t)} \define
            \sum_{m\in\esSet,j\in\jSpace}
            \Brace{&
                \sum_{k\in\apSet} \Inorm{\vec{R}^{(k)}_{m,j}(t,0)} + Q_{m,j}(t,0)
                \nonumber\\
                &~~~~~+ \beta \cdot I[Q_{m,j}(t,0)=L_{max}]
            },
    \end{align}
    }%
    where $\Inorm{\vec{x}}$ denotes the $L^1$-norm of the vector $\vec{x}$, $\beta$ is the weight of overflow penalty\comments{, and $I[\cdot]$ denotes the indicator function which is equal to $1$ when the inner statement is true and $0$ otherwise}.
\end{definition}

\comments{Then, }since the job dispatching in one broadcast interval will affect the cost of the successive broadcast intervals, we should consider the joint minimization of the cumulative costs of all the broadcast intervals.
Specifically, we consider the following discounted sum of the costs of all the broadcast intervals as the system objective.
{\small
\begin{align}
    &\bar{G}(\Stat, \Policy) \define
    \lim_{T \to \infty} \mathbb{E}^{\Policy}_{\set{\Stat(t)|\forall t}}
    \Bracket{
        \sum_{t=1}^{T} \gamma^{t-1} g\Paren{\Stat(t)} \Big| \Stat(1)
    },
\end{align}
}%
where $\mathbb{E}^{\Policy}_{\set{\Stat(t)|\forall t}}[\cdot]$ denotes the expectation with respect to all possible system states in the future given scheduling policy $\Policy$, and $\gamma \in (0,1)$ is the discount factor.
Hence, the optimization of job dispatching policy can be formulated as the following minimization problem.
{\small
\begin{align}
    \textbf{P1:}~
    \min_{\Policy} \bar{G}(\Stat, \Policy).
\end{align}
}%

\comments{Finally, }if the GSI $\Stat(t)$ and \brlatency~$\Delay(t)$ are known to all the APs,
the MDP in problem P1 can be solved via the following Bellman's equations as in \cite{sutton1998}:
{\small
\begin{align}
    &V\Paren{\Stat(t)} =g\Paren{\Stat(t)}
        + \gamma\mathbb{E}_{\Delay}\bigg\{
            \min_{\Policy(\Stat(t),\Delay(t))}
            \nonumber\\
            &\sum_{\Stat(t+1)} \Pr \Big\{ 
                \Stat(t+1) \Big| \Stat(t), \Policy(\Stat(t), \Delay(t)) \Big\} \cdot V\Big(\Stat(t+1)\Big)
            \bigg\},
    \label{eqn:sp_0}
\end{align}
}%
where the value function $V(\Stat(t))$ of the optimal policy $\Policy^{*}$ 
is defined as follows.
{\small
\begin{align}
    &V\Paren{\Stat(t)} \define
    \lim_{T\to\infty} 
    \mathbb{E}^{\Policy^*}_{\set{\Stat(t)|\forall t}} \Bracket{
        \sum_{t=1}^{T} \gamma^{t-1} g\Big( \Stat(t) \Big) \Big| \Stat(1)
    }.
    \label{eqn:val_f}
\end{align}
}%
Moreover, the optimal policy $\Omega^{*}$ can be obtained by solving the right-hand-side (RHS) of the above Bellman's equations.

However, it is infeasible to solve the above Bellman's equations because each AP only has the knowledge of its own OSI and local \brlatency~in our considered edge computing system.
Thus problem P1 is actually a POMDP, whose general solution is of huge complexity \cite{IJCAI03-NairR,IJCAI99-BoutilierC}.

    \section{Distributed Algorithm with Partial Information}
\label{sec:algorithm}
In this section, we shall introduce a novel approximation method to decouple the centralized optimization on the RHS of the Bellman's equations in equation (\ref{eqn:sp_0}) to each AP for arbitrary system state.
For the APs outside the conflict AP sets of each other, the update of dispatching actions at one AP will not affect the task computation originated from other APs.
On the other hand, for the APs within the same conflict AP set, the optimization of their dispatching actions is coupled.
Hence, we introduce an alternative actions update algorithm to optimize the dispatching actions of subset of $\apSet$ in each broadcast interval periodically, while other APs maintain their dispatching actions in the previous broadcast interval.
Specifically, the proposed distributed algorithm consists of the following two steps:
\begin{enumerate}
    \item We first introduce a time-variant baseline policy, use its value function to approximate the value function of the optimal policy $\Policy^*$ in each broadcast interval, and derive the analytical expression of the approximate value function for arbitrary GSI in Section \ref{subsec:baseline}.
    \item With the approximate value function, an alternative actions update algorithm, where only a subset of APs are selected to update their dispatching actions distributedly in each broadcast interval, is proposed in Section \ref{subsec:ap_alg}.
    Moreover, the analytical performance bound is derived in Section \ref{subsec:analysis}.
\end{enumerate}

\subsection{Baseline Policy and Approximate Value Function}
\label{subsec:baseline}
To alleviate the curse of dimensionality, we first use the baseline policy with fixed dispatching actions to approximate the value function at the RHS of the Bellman's equations in equation (\ref{eqn:val_f}).
The baseline policy is elaborated below.

\begin{policy}[Baseline Policy]
    In the baseline policy $\Baseline$, each AP fixes its dispatching actions as in the previous broadcast interval. Specifically, at the $t$-th broadcast interval,
    {\small
    \begin{align}
        \Baseline\Paren{\Stat(t),\Delay(t)} &\define \Brace{ 
            \Pi_{1}(\Stat_{1}(t),\mathcal{D}_{1}(t)),
            \dots,
            \Pi_{K}(\Stat_{K}(t),\mathcal{D}_{K}(t))
        },
    \end{align}
    }%
    where
    {\small
    \begin{align}
        \Pi_{k}\Paren{\Stat_{k}(t),\mathcal{D}_{k}(t)}
        &\define \Brace{
            {\omega}_{k,j}(t+1) \Big| \forall j\in\jSpace
        }, \forall k\in\apSet.
    \end{align}
    }%
\end{policy}

Let $W_{\Baseline}(\cdot)$ be the value function of the baseline policy $\Baseline$, we shall approximate the value function of the optimal policy $V(\cdot)$ via $W_{\Baseline}$, i.e.,
{\small
\begin{align}
    &V\Paren{\Stat(t+1)} \approx W_{\Baseline}\Paren{\Stat(t+1)}
    \nonumber\\
    =& \sum_{m\in\esSet,j\in\jSpace}\Brace{
        \sum_{k\in\apSet} \tilde{W}^{\AP}_{k,m,j}(\Stat(t+1))
        +\tilde{W}^{\ES}_{m,j}(\Stat(t+1))
    },
    \label{eqn:baseline}
\end{align}
}%
where $\tilde{W}^{\AP}_{k,m,j}(\Stat(t+1))$ denotes the cost raised by the type-$j$ jobs which are being transmitted from the $k$-th AP to the $m$-th processing server with the baseline policy $\Baseline$ and initial system state $\Stat(t+1)$, and $\tilde{W}^{\ES}_{m,j}(\Stat(t+1))$ denotes the cost raised by the type-$j$ jobs on the $m$-th server.
Their definitions are given below.
{\small
\begin{align}
    \tilde{W}^{\AP}_{k,m,j} \Paren{\Stat(t+1)} &\define
        \sum_{i=0}^{\infty} \gamma^{i+1} \mathbb{E}^{\Baseline}\Bracket{
            \Inorm{\vec{R}^{(k)}_{m,j}(t+i+1)}
        },
    \label{w_ap}
    \\    
    \tilde{W}^{\ES}_{m,j} \Paren{\Stat(t+1)} &\define
        \sum_{i=0}^{\infty} \gamma^{i+1} \mathbb{E}^{\Baseline}\Bracket{
            Q_{m,j}(t+i+1) +
            \nonumber\\
            &~~~~~~~~~~\beta I[Q_{m,j}(t+i+1) = L_{max}]
        }.
    \label{w_es}
\end{align}
}%

\subsection{Distributed Actions Update}
\label{subsec:ap_alg}
Although the optimal value function has been approximated via the baseline policy in the previous part, it is still infeasible to solve the RHS of the Bellman's equations as the evaluation of equation (\ref{w_ap}) and (\ref{w_es}) requires the knowledge of GSI and \brlatency~at all APs.
Instead, it is feasible for part of APs to update their dispatching actions distributedly in each broadcast interval and achieve a better performance compared with baseline policy.
Hence, we first define the following sequence of AP subsets, where each subset are selected to update dispatching actions periodically.
\begin{definition}[Subset Partition]
    Let $\mathcal{Y}_{1}, \dots, \mathcal{Y}_{N} \subseteq \apSet$ be a collection of subsets of AP set $\apSet$, which satisfy
    {\small
    \begin{align}
        &\bigcup_{n=0,\dots,N-1} \mathcal{Y}_{n} = \apSet
        \label{eqn:subset_cup}
        \\
        \esSet_{y} \cap \esSet_{y'} &=\emptyset, y' \neq y~(\forall y',y \in \mathcal{Y}_{n}).
        \label{eqn:subset_disjoint}
    \end{align}
    }%
\end{definition}
The subset partition is not trivial and a partition to minimize the update period $N$ is preferred.
A heuristic greedy algorithm is given in Algorithm \ref{alg_0}.
\begin{algorithm}[ht]
    \caption{Greedy Subset Partition Algorithm}\label{alg_0}
    \DontPrintSemicolon 
    \KwIn{$\apSet, \set{\esSet_{k}, \forall k\in\apSet}$ }
    \KwOut{a subset partition $\set{ \mathcal{Y}_{n} }$}
    Initialize a subset partition as $\mathcal{Y}_{n} = \set{n}$ ($\forall n\in\apSet$).\;
    \While{ $\exists \mathcal{Y}_a$ and $\mathcal{Y}_b$ ($a \neq b$) that $\cup\set{ \esSet_{y}|y\in\mathcal{Y}_a } \bigcap \cup\set{ \esSet_{y}|y\in\mathcal{Y}_b } \neq \emptyset$ }
    {
        Count number of subsets in the current subset partition which have disjoint candidate set with $\mathcal{Y}_n$ ($\forall n$), denoted the number as $I_{n}$.\;
        $\tilde{n} \gets \arg\min_{n} I_{n}$\;
        Merge subset $Y_{\tilde{n}}$ with one of its disjoint subsets.\;
    }
\end{algorithm}

At the $t$-th broadcast interval, the APs in the subset indexed with $n \define t \pmod{N}$ update their dispatching actions, while the other APs keep the same dispatching actions as the previous broadcast interval.
Hence, let
{\small
\begin{align}
    \tilde{\mathcal{A}}(t) \define \Brace{ {\omega}_{y,j}(t+1) \Big| \forall y\in\mathcal{Y}_{n},j\in\jSpace }
\end{align}
}%
be the aggregation of dispatching actions of the APs in the subset $\mathcal{Y}_{n}$, and
{\small
\begin{align}
    \hat{\mathcal{A}}(t) \define \Brace{ {\omega}_{y,j}(t+1) \Big| \forall y\notin\mathcal{Y}_{n}, j\in\jSpace}
\end{align}
}%
be the aggregation of dispatching actions of the rest APs, which are same as in the previous broadcast interval.
At the $t$-th broadcast interval, the optimization of dispatching actions $\tilde{\mathcal{A}}(t)$ at the RHS of the Bellman's equations can be rewritten as the following problem.
{\small
\begin{align}
    \textbf{P2:}~
    \min_{ \tilde{\mathcal{A}}(t) }
    &\sum_{\Stat(t+1)} \Pr\Brace{
        \Stat(t+1) \Big| \Stat(t), \hat{\mathcal{A}}(t), \tilde{\mathcal{A}}(t)
    } \cdot W_{\Baseline}\Paren{\Stat(t+1)}.
\end{align}
}%

Moreover, we have the following conclusion on the decomposition of P2.
\begin{lemma}[]
    The optimization problem in P2 can be equivalently decoupled into local optimization problems at APs {for each subset partition}.
    Specifically, {the local optimization for the $y$-th AP in the $n$-th subset ($\forall n$) can be written as}
    {\small
    \begin{align}
        &\textbf{P3:}~
        \min_{ {\mathcal{A}}_{y}(t+1) }
        \mathbb{E}_{\set{ \Stat_{y}(t+1)|\Stat_{y}(t), \hat{\mathcal{A}}(t), {\mathcal{A}}_{y}(t+1) }}
        \nonumber\\
        &~~~~\sum_{ j\in\jSpace,m\in\esSet_{y} } \Brace{
            \tilde{W}^{\AP}_{k,j}\Paren{\Stat_{y}(t+1)}
            +\tilde{W}^{\ES}_{m,j}\Paren{\Stat_{y}(t+1)}
        }.
        \label{eqn:partial}
    \end{align}
    }%
    \label{lemma:w_partial}
\end{lemma}
\begin{proof}
    At the $t$-th broadcast interval, the $y$-th AP in the subset $\mathcal{Y}_{n}$ updates its dispatching actions, which could only affect the future cost raised on itself and its corresponding \emph{candidate server set}, i.e., the part of its OSI.
    Hence, it's obvious that the equation (\ref{w_ap}) and equation (\ref{w_es}) on the RHS of the Bellman's equations could be reduced into the form based only on the OSI of the $y$-th AP ($\forall y\in\mathcal{Y}_{n}$).
\end{proof}

The optimization of {dispatching actions $\mathcal{A}_{y}(t+1)$} for the $y$-th AP ($\forall y\in\mathcal{Y}_{n}$) in P3 could be achieved via searching all the processing servers in $\esSet_{y}$, {whose computational complexity is $O(J|\mathcal{M}_{y}|)$}.
As a result, the overall algorithm of job dispatching is elaborated in Algorithm \ref{alg_1}.
\begin{algorithm}[ht]
    \caption{{Online Alternative Actions Update Algorithm}}\label{alg_1}
    \DontPrintSemicolon 
    Initialize all the APs with heuristic dispatching actions $\set{{\omega}_{k,j}(0)|\forall k\in\apSet,j\in\jSpace}$.\;
    \For{$t=0,1,2,\dots$}{
        $n \gets t \pmod{N}$\;
        \ForPar{$y \in \mathcal{Y}_{n}$}{
            The $y$-th AP observes $\Stat_{y}(t)$ after $\mathcal{D}_{y}(t)$.\;
            Solve P3 with $\Stat_{y}(t), \mathcal{D}_{y}(t)$ and obtain optimized actions $\set{\tilde{\omega}_{y,j}(t+1)|\forall j\in\jSpace}$\;
        }
        $\tilde{\mathcal{A}}(t+1) \gets \set{\tilde{\omega}_{y,j}(t+1)|\forall y\in\mathcal{Y}_{n},j\in\jSpace}.$\;
        $\hat{\mathcal{A}}(t+1) \gets \hphantom{~~} \set{ {\omega}_{y,j}(t) | \forall y\in\mathcal{Y}_{n-1},j\in\jSpace }$\;
        $\hphantom{~~~~~~~~~~~~~~~~~} \cup \set{ {\omega}_{y,j}(t) | \forall y\notin\mathcal{Y}_{n-1},j\in\jSpace }$\;
    }
\end{algorithm}

As a remark notice that since in Algorithm \ref{alg_1}, the computation complexity at each AP scales linearly with respect to {the size of candidate edge server set}, it can be deployed in a scenario with massive APs and edge servers, as long as the {the available number of edge servers for each AP} is limited.

\subsection{Analytical Performance Bound}
\label{subsec:analysis}
In most of the existing approximate MDP solutions \cite{mdp-bound1,mdp-bound2,mdp-bound3,mdp-bound4,mdp-bound5,mdp-bound6}, the performance is difficult to bound analytically as the approximate value function has no accurate meaning on the system cost or utility.
In the proposed algorithm, however, we derive the analytical expression for the baseline policy as the approximate.
Hence, the alternative dispatching actions update can ensure to achieve a better performance than the baseline policy.
This conclusion is summarized below.
\begin{lemma}[Analytical Cost Upper Bound]
    \label{lemma:bound}
    Let $W_{\tilde{\Policy}}(\cdot)$ be the value function (average cost) of the proposed policy $\tilde{\Omega}$, i.e.,
    {\small
    \begin{align}
        W_{\tilde{\Policy}}(\Stat) \define
        \sum_{t=1}^{\infty} \gamma^{t-1} \mathbb{E}^{ \tilde{\Policy} }_{\set{\Stat(t)|\forall t}} \Bracket{
            g\Paren{\Stat(t)} \Big| \Stat(1)=\Stat
        },  
    \end{align}
    }%
    we have
    {\small
    \begin{align}
        V(\Stat)
        \leq W_{\tilde{\Policy}}(\Stat)
        \leq W_{\Baseline}(\Stat),
        \forall \Stat.
    \end{align}
    }%
\end{lemma}
\begin{proof}
    $V(\Stat) \leq W_{\tilde{\Policy}}(\Stat)$ is straightforward as $\tilde{\Policy}$ is not optimal policy.
    The proof of $W_{\tilde{\Policy}}(\Stat) \leq W_{\Baseline}(\Stat)$ is equivalent to prove the improvement of one-step policy iteration, which is similar to the proof of \emph{Policy Improvement Property} in \cite{dp-control}.
\end{proof}
Therefore, $W_{\Baseline}(\Stat)$ derived in equation (\ref{eqn:baseline}) can be used as the analytical cost upper bound of the proposed policy $\Baseline$.
Moreover, Lemma \ref{lemma:bound} also implies that the proposed policy with fixed service edge server for each AP, as long as the {static dispatching policy} is used as the baseline policy.

    \section{Performance Evaluation}
\label{sec:evaluation}
In this section, we evaluate the performance of the proposed low-complexity dispatching policy $\tilde{\Policy}$ by numerical simulations.
The experiment setup and performance benchmarks are elaborated in Section \ref{subsec:setup}.
The simulation results are illustrated in Section \ref{subsec:basic}.
The sensitivity study on parameters is also applied to provide some insights on the robustness of the proposed policy in Section \ref{subsec:advance}.

\begin{figure}[t!]                                                      %
    \centering                                                          %
    \includegraphics[width=0.42\textwidth]{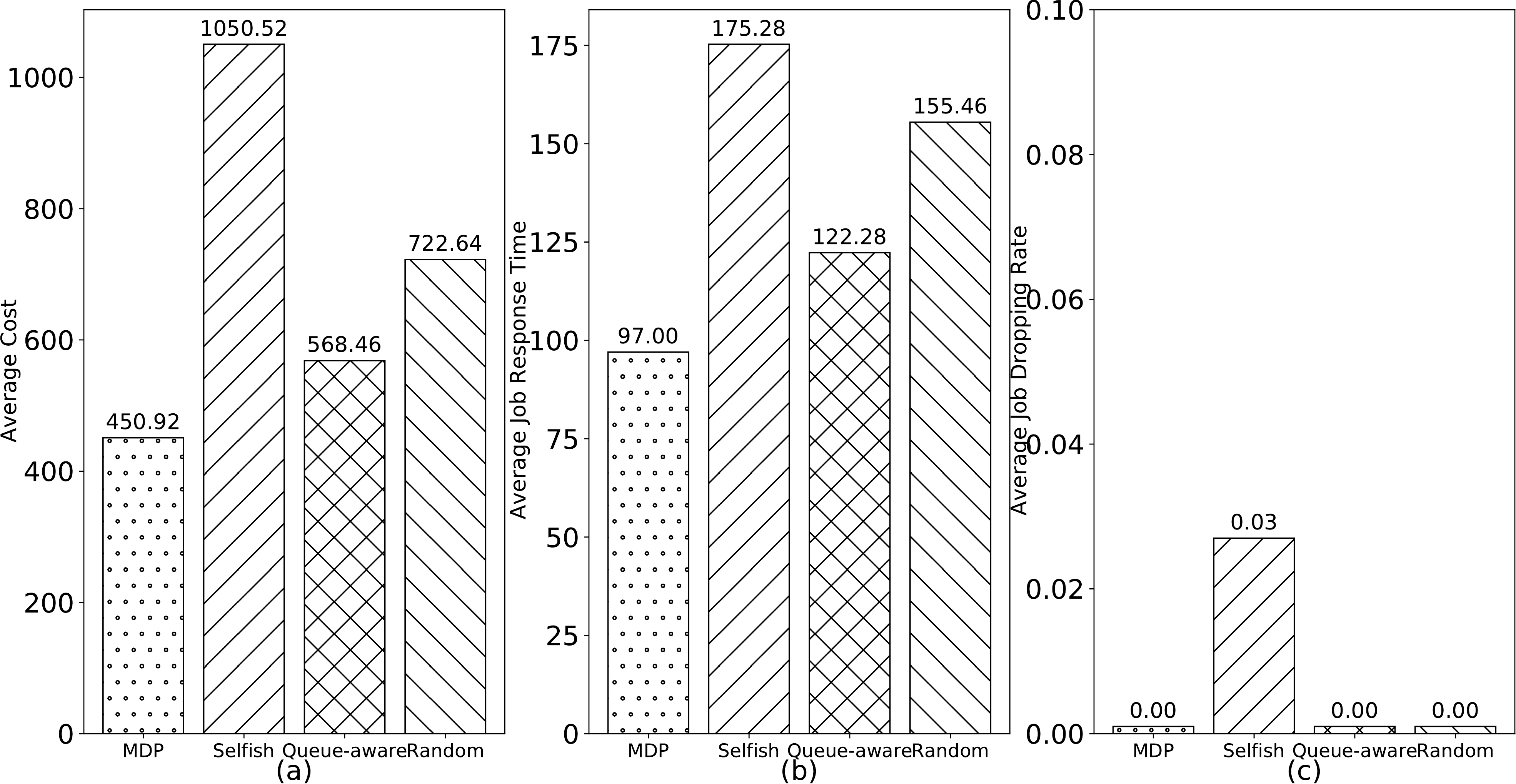}               %
    \caption{Illustration of performance metrics comparison with benchmarks.}
    \label{fig:bar_plot}                                                %
\end{figure}                                                            %

\begin{figure}[t!]                                                                             %
    \centering                                                                                  %
    \includegraphics[width=0.42\textwidth]{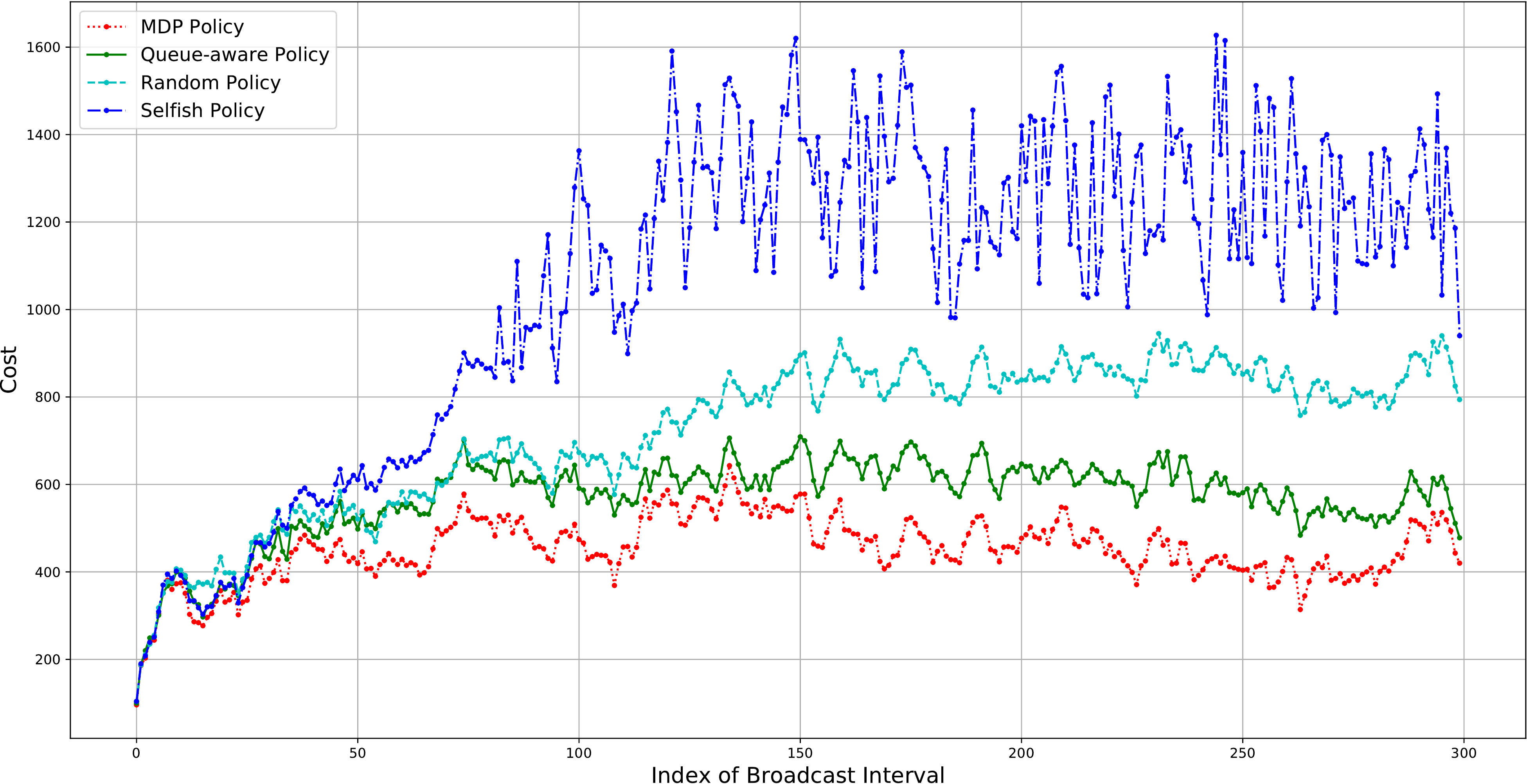}                     %
    \caption{Illustration of cost versus index of broadcast interval.}
    \label{fig:general_timeline}                                                                %
\end{figure}                                                                                    %

\subsection{Experiment Setup}
\label{subsec:setup}

\begin{figure*}[t!]                                                                %
    \centering                                                                      %
    \begin{minipage}[b]{0.30\textwidth}                                             %
        \includegraphics[width=\textwidth]{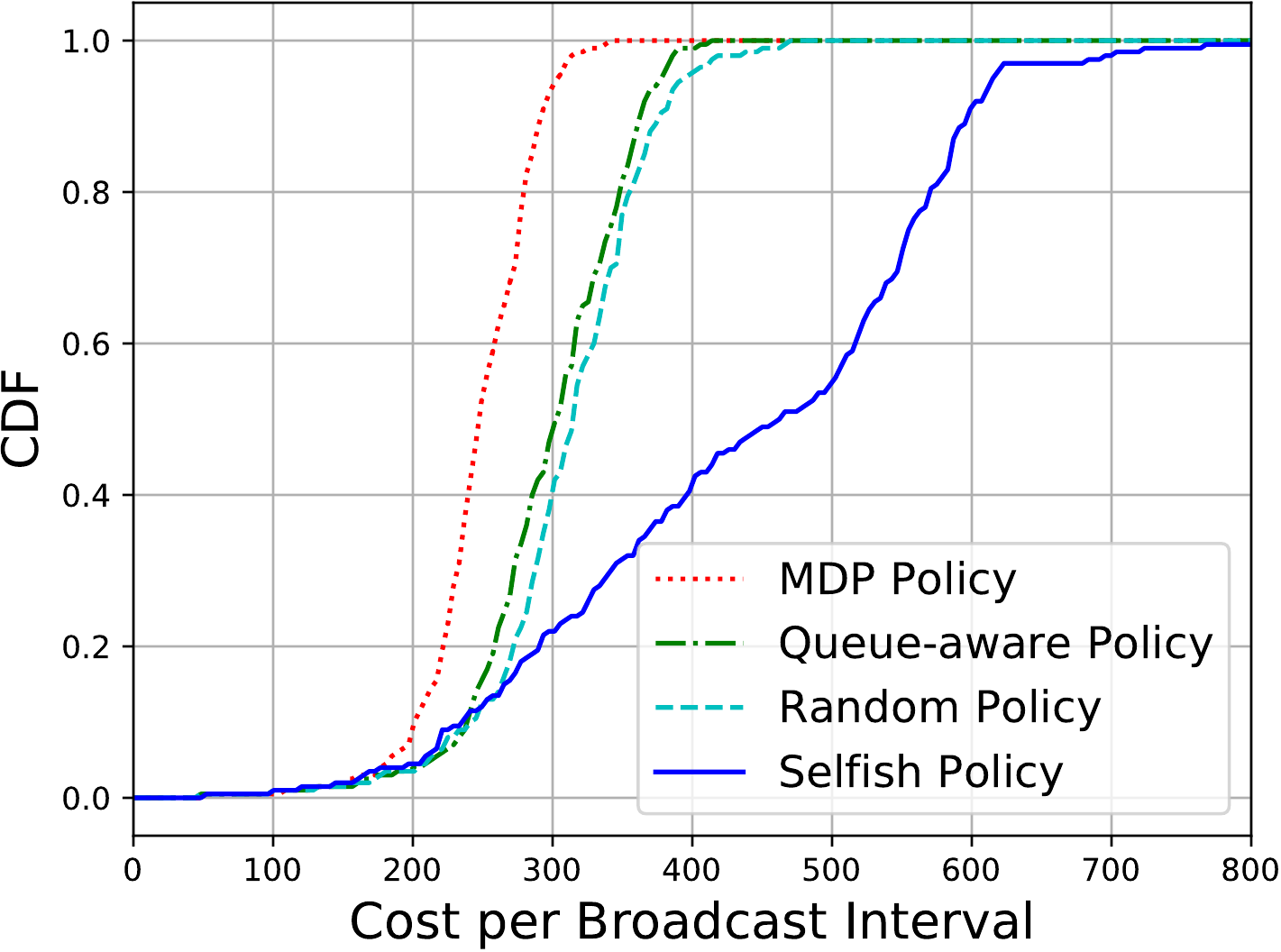} \\                  %
        {(a) \brlatency~as $5$ time slots.}                                               %
    \end{minipage}                                                                  %
    \begin{minipage}[b]{0.30\textwidth}                                             %
        \includegraphics[width=\textwidth]{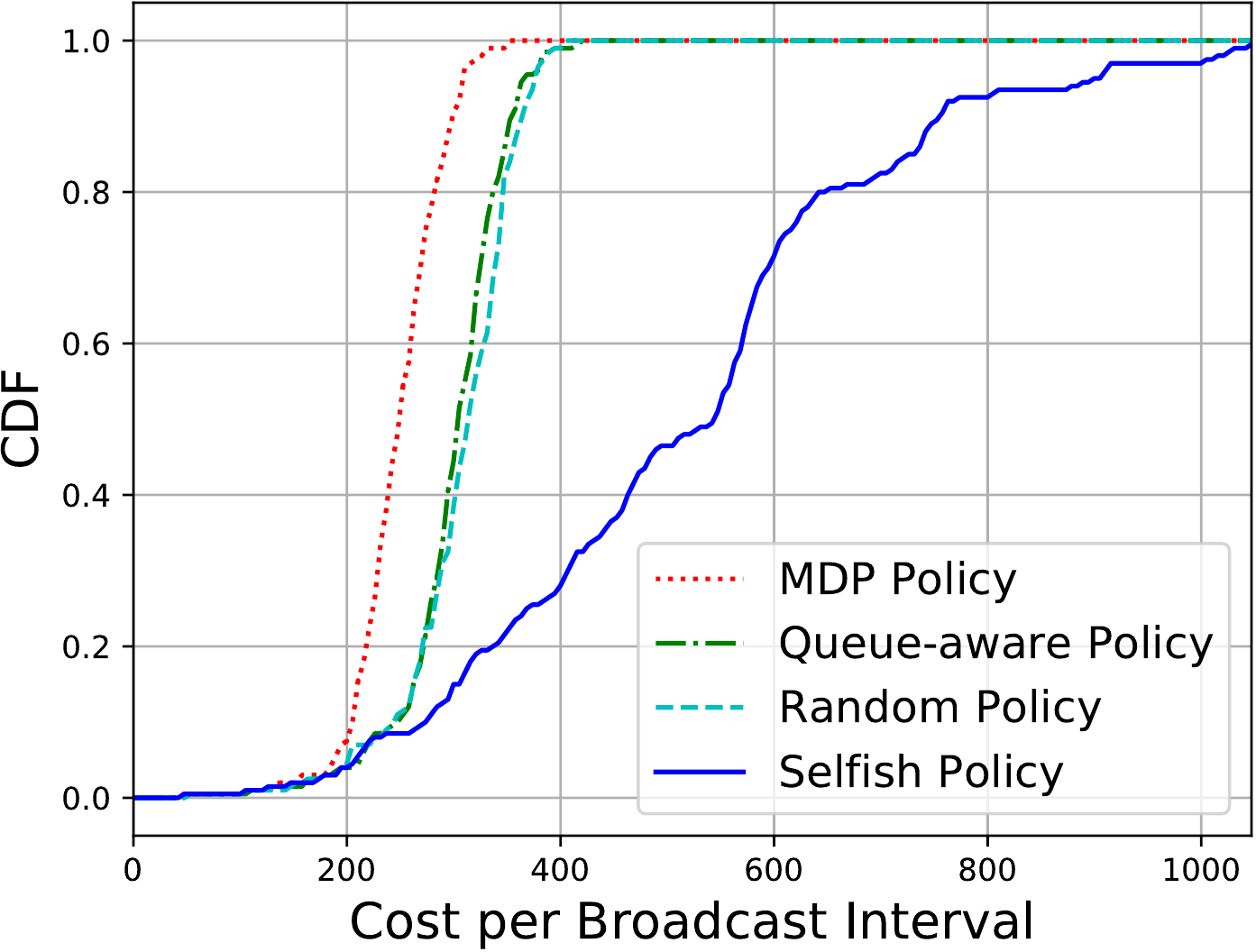} \\                 %
        {(b) \brlatency~as $12$ time slots.}           %
    \end{minipage}                                                                  %
    \begin{minipage}[b]{0.30\textwidth}                                             %
        \includegraphics[width=\textwidth]{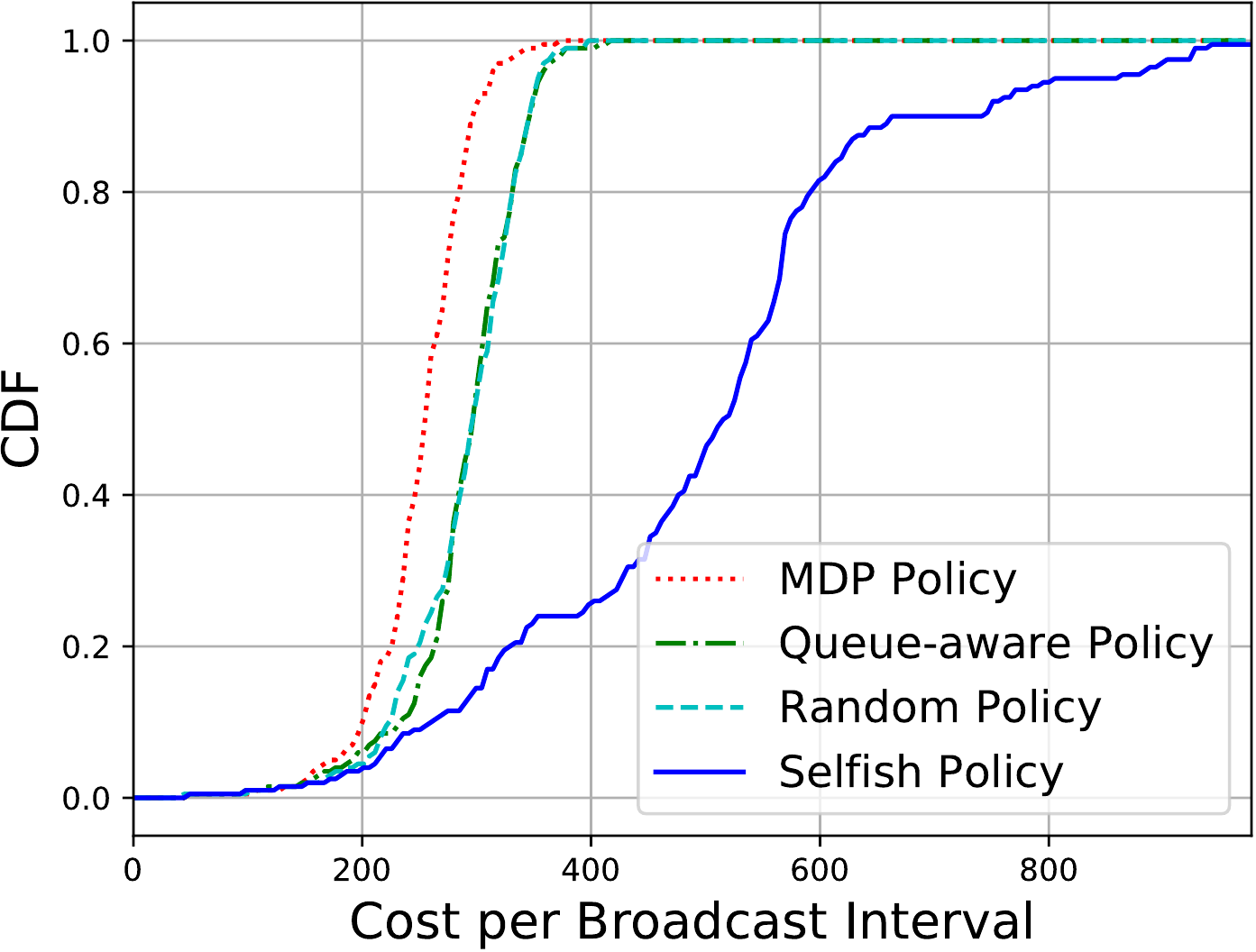} \\                  %
        {(c) \brlatency~as $25$ time slots.}             %
    \end{minipage}                                                                  %
    \caption{Algorithm Robustness versus Signaling Latency.}                        %
    \label{fig:ss_signal}                                                           %
\end{figure*}                                                                       %

In the simulation, we assume that there are $K=15$ APs, $M=10$ processing servers and $J=10$ types of jobs in the system.
One broadcast interval is consist of $t_{B}=25$ time slots.
The network topology among APs is generated according to Barab\'asi-Albert (BA) model \cite{albert1999diameter} and the processing servers are randomly placed collocated with the APs.
The arrival traces and job processing time for each job type are extracted from Google cluster traces \cite{clusterdata:Reiss2011} and then randomly assigned on APs and edge servers, respectively.
The maximum uploading latency is $\Xi = 3t_B$, and the distribution of $\mathbb{U}_{k,m,j}(\Xi)$ ($\forall k\in\apSet, m\in\esSet_{k}, j\in\jSpace$) is arbitrarily generated within the support $\set{0, 1, \dots, \Xi}$.
The \brlatency~ is with an integer support from $0.7t_B$ to $0.9t_B$ time slots.
Each queue for VMs on edge server is with maximum queue length $L_{max}=50$, i.e., there would be at most $50$ jobs on one edge server.
The discount factor $\gamma$ is $0.95$ and the overflow penalty $\beta$ is $120$.


We also propose {three heuristic benchmarks to profile the performance of the proposed MDP policy}. 
\begin{itemize}
    \item \textbf{Random Dispatching Policy}:
            Randomly choose a dispatching edge server in each time slot; 
    \item \textbf{Selfish Policy}:
            Always choose the edge server with the minimum sum of the expected uploading time and processing time;
    \item \textbf{Queue-aware Policy}:
            Always choose the edge server with the minimum sum of expected uploading time, processing time and queueing time based on the observation of outdated queue states.
\end{itemize}
Moreover, we choose the Selfish Policy as the initial dispatching actions for our proposed algorithm (Algorithm \ref{alg_1}).

\subsection{Performance Analysis}
\label{subsec:basic}
As illustrated in Fig.\ref{fig:bar_plot}(a), the proposed algorithm (MDP Policy) outperforms all the benchmarks in the average system cost.
More insights on the performance comparison are provided in Fig.\ref{fig:bar_plot}(b) and (c).
In the former figure, the average job response times, measuring the average number of broadcast intervals from job's arrival at one AP to the completeness of computation at one edge server, are compared.
It can be observed that the proposed policy still outperforms all the benchmarks.
In Fig.\ref{fig:bar_plot}(c), the job dropping rates, measuring the ratio of jobs dropped by edge servers due to queue overflow, are also compared.
It is shown that the proposed policy outperforms other three benchmarks with the minimum average cost and job response time.
And there is no dropping jobs incurred compared with the Selfish policy, which is the initial baseline policy for our proposed algorithm.
Finally, an realization of job dispatching is illustrated in Fig.\ref{fig:general_timeline}, where the number of jobs in the system is plot versus the index of broadcast interval.
It can be observed that the proposed policy manage to keep the number of jobs in lower level, compared with the other benchmarks.
This demonstrates its high dispatching efficiency.

\begin{figure}[ht!]
    \centering
    \begin{minipage}[b]{0.23\textwidth}
        \includegraphics[width=\textwidth]{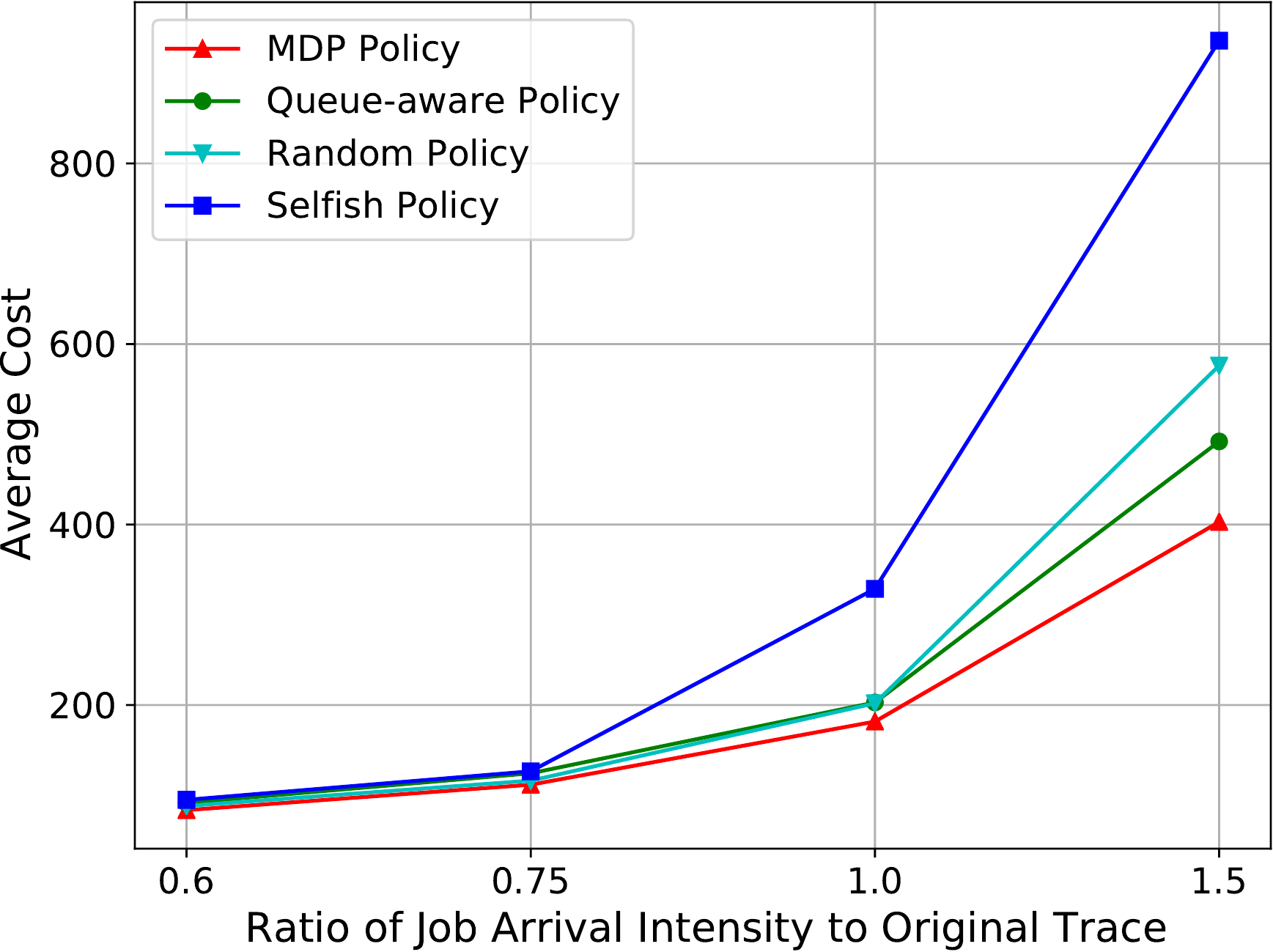}
        \caption{Illustration of average system cost versus job arrival intensity.}
        \label{fig:ss_scale}
    \end{minipage}
    \begin{minipage}[b]{0.23\textwidth}
        \includegraphics[width=\textwidth]{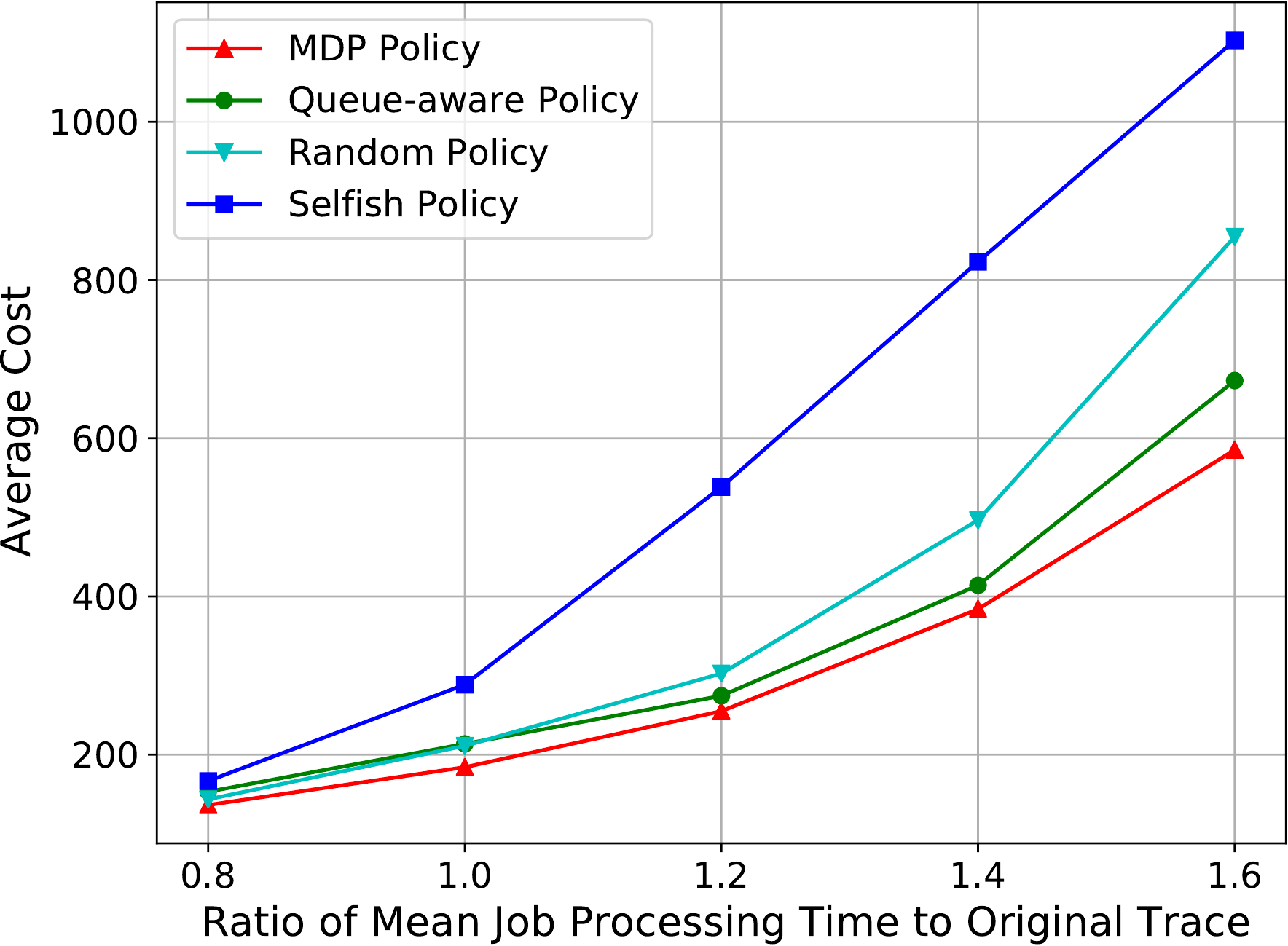}
        \caption{Illustration of average system cost versus mean processing time.}
        \label{fig:ss_dist}
    \end{minipage}
\end{figure}

\subsection{Sensitivity Study}
\label{subsec:advance}

\noindent\textbf{Signaling Latency.}
The simulation results with different \brlatency~$\mathcal{D}_{k}$ ($\forall k\in\apSet$) are illustrated in Fig.\ref{fig:ss_signal}, where the cumulative distribution function (CDF) of the job number in the system is plotted.
Specifically, the \brlatency~of all the APs is set to $5, 12, 25$ in Fig.\ref{fig:ss_signal}(a), Fig.\ref{fig:ss_signal}(b), Fig.\ref{fig:ss_signal}(c), respectively.
It can be observed that with the increasing of \brlatency, the performance of Queue-aware Policy becomes worse. 
The Queue-aware policy slightly outperforms the Random Policy in Fig.5(a) with smaller \brlatency~(achieving a smaller number of jobs in the system), and becomes worse in Fig.5(c) with large \brlatency.
This demonstrates that the Queue-aware Policy is sensitive to \brlatency.
In all the figures, the proposed policy outperforms all the benchmarks, which demonstrates its robustness versus signaling latency.

\noindent\textbf{Job Arrival Intensity.}
We carry out the sensitivity study of job arrival intensity by integer scaling the interval of jobs arriving in Google cluster traces.
The average system cost versus the number of APs is illustrated in Fig.\ref{fig:ss_scale}.
With the increasing of job arrival intensity, the average system cost increases in all the benchmarks and our proposed policy.
It can be observed that our policy performs the best.
Moreover, the performance gain becomes significant when the computation load is heavy.

\noindent\textbf{Mean Processing Time.}
The simulation results of various mean processing time are illustrated in Fig.\ref{fig:ss_dist}, where the mean processing time is taken as $c_{m,j}$ of the processing time distribution $\mathbb{G}(1/c_{m,j})$ in our computation model assumption.
Generally speaking, with the increasing average processing time, the average system cost increases in all the benchmarks and our proposed policy.
The simulation results are consistent with that in Fig.\ref{fig:ss_scale}.
\delete{temp}{
    It can be observed that the proposed policy has better performance than the benchmarks.
    Moreover, the performance gain becomes significant when the computation time is long.
}


    \section{Conclusion}
\label{sec:conclusion}
In this paper, we consider an online distributed job dispatcher design problem for an edge computing system residing in a Metropolitan Area Network.
In this edge computing system, the job dispatchers are implemented in a distributed manner on multiple access points (APs) which collect jobs from mobile users and then dispatch jobs to one edge server or cloud server for processing.
To facilitate the cooperation among distributed job dispatchers, a signaling mechanism is introduced where the APs and edge servers would periodically broadcast their local state information to the job dispatchers.
However, the reception of updated and fully-observed global system state is discouraged as the transmission latency is non-negligible in MAN and reception of all broadcast is time consuming.
Hence, we formulate the distributed optimization problem of job dispatching strategies as a POMDP problem, with outdated and partially-observable information.
The conventional solution for POMDP is impractical due to huge time complexity.
In this paper, we propose a novel low-complexity solution framework for distributed job dispatching, based on which the optimization of job dispatching policy can be decoupled via an alternative policy iteration algorithm and a theoretical performance lower bound is obtained.
The evaluation results show that our proposed policy can achieve obvious and robust performance gain compared with heuristic baselines.
\comments{
Furthermore, this work assumes available knowledge on the distributions of signaling latency, uploading latency and computation time.
As an extension, the reinforcement learning could be integrated with the proposed solution framework when the above statistics are absent.
}

    \ifCLASSOPTIONcaptionsoff
        \newpage
    \fi

    \bibliographystyle{IEEEtrans}
    \bibliography{final.bib}

\end{document}